\documentclass[twocolumn,aps,prx,groupedaddress,showpacs,superscriptaddress,floatfix]{revtex4-2} 

\usepackage{amsmath}
\usepackage{amsfonts}
\usepackage{upgreek}
\usepackage{amssymb}
\usepackage{graphicx,bm,units,yfonts}
\usepackage[table]{xcolor}
\usepackage{color}

\usepackage[utf8]{inputenc}

\usepackage{amstext}

\makeatletter
\let\openbox\@undefined
\makeatother
\usepackage{amsthm}

\makeatletter

\theoremstyle{plain}

\theoremstyle{remark}

\renewcommand{\vec}[1]{\mathbf{#1}}
\newcommand{\trans}{\mathsf{T}}

\makeatother

\providecommand{\remarkname}{Remark}
\providecommand{\theoremname}{Theorem}

\begin{document}

\title{An operator-based approach to topological photonics}

\author{Alexander Cerjan}
\email{awcerja@sandia.gov}
\affiliation{Center for Integrated Nanotechnologies, Sandia National Laboratories, Albuquerque, New Mexico 87185, USA}

\author{Terry A.\ Loring}
\email{loring@math.unm.edu}
\affiliation{Department of Mathematics and Statistics, University of New Mexico, Albuquerque, New Mexico 87131, USA}

\date{\today}

\begin{abstract}
Recently, the study of topological structures in photonics has garnered significant interest, as these systems can realize robust, non-reciprocal chiral edge states and cavity-like confined states that have applications in both linear and non-linear devices. However, current band theoretic approaches to understanding topology in photonic systems yield fundamental limitations on the classes of structures that can be studied. Here, we develop a theoretical framework for assessing a photonic structure's topology directly from its effective Hamiltonian and position operators, as expressed in real space, and without the need to calculate the system's Bloch eigenstates or band structure. Using this framework, we show that non-trivial topology, and associated boundary-localized chiral resonances, can manifest in photonic crystals with broken time-reversal symmetry that lack a complete band gap, a result which may have implications for new topological laser designs. Finally, we use our operator-based framework to develop a novel class of invariants for topology stemming from a system's crystalline symmetries, which allows for the prediction of robust localized states for creating waveguides and cavities.
\end{abstract}

\maketitle

\section{Introduction}

In recent years, the incorporation of concepts from topological physics into photonic structures has generated significant interest \cite{lu_topological_2014,khanikaev_two-dimensional_2017,ozawa_topological_2019,xie_higher-order_2021}, as such systems can realize robust, localized states for enhancing light-matter interactions \cite{ota_active_2020} and routing quantum information \cite{lodahl_chiral_2017}. For example, topological lasers \cite{bahari_nonreciprocal_2017,st-jean_lasing_2017,bandres_topological_2018,zeng_electrically_2020,yang_spin-momentum-locked_2020,shao_high-performance_2020,bahari_photonic_2021,dikopoltsev_topological_2021,yang_topological-cavity_2022} can exhibit efficient phase locking and increased robustness against disorder in comparison to their conventional counterparts \cite{harari_topological_2018,amelio_theory_2020,zapletal_long-lived_2020}, while many topological photonic systems, such as those relying on the photonic analogues to the valley-Hall and quantum spin Hall effects, can be used to create, direct, and protect quantum states \cite{rechtsman_topological_2016,barik_topological_2018,mittal_topological_2018,barik_chiral_2020,parappurath_direct_2020,arora_direct_2021,dai_topologically_2022,hauff_chiral_2022}. However, these developments have also exposed some of the current fundamental challenges facing the field of topological photonics. First, it has proven difficult to realize nanoscale photonic Chern insulators, and their highly desirable non-reciprocal chiral edge states, using magneto-optic materials; known materials have relatively weak magneto-optic responses at technologically relevant wavelengths \cite{espinola_magneto-optical_2004,bi_-chip_2011}, which makes it hard to use this response to open a complete topological band gap in a photonic crystal. While there are many photonic Chern insulator designs that circumvent this problem by using alternate methods for breaking (or effectively breaking) time-reversal symmetry \cite{wang_observation_2009,rechtsman_photonic_2013,hafezi_imaging_2013,khanikaev_photonic_2013,wu_scheme_2015,wimmer_experimental_2017,zilberberg_photonic_2018,klembt_exciton-polariton_2018,lustig_photonic_2019,fang_anomalous_2019,liu_generation_2020,dutt_single_2020,lu_floquet_2021}, these methods incur costs in increased system size or experimental complexity. Second, even in topological photonic systems that preserve time-reversal symmetry, such as photonic topological crystalline insulators \cite{blanco-redondo_topological_2016,noh2018,peterson2018,blanco-redondo_topological_2018,chen_direct_2019,xie_visualization_2019,mittal_photonic_2019,ota_photonic_2019,smirnova_third-harmonic_2019,cerjan2020_HOT_BIC_exp,kim_multipolar_2020,kruk_nonlinear_2021}, we do not have a general theory for how to treat their topological properties in finite systems. Indeed, many potential applications of the cavity-like states found in higher-order topological photonic systems would benefit from coupling these effective cavities together. However, at present, finite-size effects in these structures must be understood using bespoke analyses of individual systems. Finally, theories of topology for non-linear photonic systems \cite{lumer_self-localized_2013,ablowitz_linear_2014,leykam_edge_2016,zhou_optical_2017,mukherjee_observation_2020,maczewsky_nonlinearity-induced_2020,leykam_probing_2021,mittal_topological_2021,jurgensen_quantized_2021,jurgensen_chern_2022,fu_nonlinear_2022,maluckov_nonlinear_2022} are also tailored to specific system architectures and are difficult to generalize; as the non-linearity generally breaks these systems' crystal symmetries, band theoretic approaches cannot be directly applied without substantial alterations.

Altogether, these fundamental challenges in the field of topological photonics can be abstracted to limitations of band theories of topology: Incorporating finite-size effects and non-linearities requires significant augmentation of a band theory, possibly to use a large supercell, while relying on a bulk band gap to be the measure of a system's topological protection forces these theories to only search for topology in insulators and semi-metals. Instead, a theory of topological photonics that is independent of band theory would potentially provide a path towards solving or circumventing this myriad of challenges currently facing the field.

Here, we develop a theoretical framework for determining a photonic structure's topology from its real-space operator description (i.e., Maxwell's equations), without the need to calculate its band structure or Bloch eigenstates. Instead, our framework is based on the photonic system's \textit{spectral localizer}, which yields a set of local invariants, protected by local gaps, for every symmetry class. Using this framework, we predict that non-trivial topology, and associated boundary-localized resonances, can appear in photonic crystals that lack complete band gaps; a discovery that may have implications for the development of photonic Chern devices at optical and near-infrared wavelengths where it is difficult to find materials that enable sufficiently strong time-reversal symmetry breaking to open complete topological photonic band gaps. Finally, we show how to develop local invariants for topology stemming from a system's crystalline symmetries and we provide an example of such an invariant for inversion symmetric photonic systems. This previously unknown class of invariants for topological crystalline structures allows for the prediction of spatially-localized topological states without the construction of Wannier centers or the calculation symmetry indicator invariants, and we anticipate that these local crystalline invariants will be useful in the development and optimization of new topologically robust photonic waveguides and cavities for enhancing light-matter interactions and routing quantum information.

\section{Theoretical Framework}

\subsection{Overview of the spectral localizer}

Over the last few years, the spectral localizer has emerged as a method for determining a finite lattice's topology directly from its real-space description using developments from the study of operator theory and $C^*$-algebras \cite{loringPseuspectra,LoringSchuBa_even,Doll_Sculz-B_skew_localizer}. There are two important conceptual shifts for defining the topology of finite lattices that distinguish such theories from traditional band theoretic approaches.
First, the system's topology can be defined locally, not globally; thus, these theories can be immediately applied to disordered and aperiodic structures without alteration, and can discriminate between different topological domains within a system. Second, as the lattice is finite (with open boundaries), it does not possess a band structure; thus, the measure of the system's topological protection determined using its real-space description must also be local. Here, it may seem intuitive to try to replace an infinitely periodic lattice's bulk band gap with a measure of protection based on a gap in the full spectrum of the finite lattice --- however, this cannot work, as boundary-localized states (of either topological or trivial origin) can close or obscure this gap, and attempting to remove these states from the full spectrum results in an inherently local measure of the protection.

At its core, the spectral localizer takes an ``operator-based'' perspective of a material's topology, in contrast to the ``eigenstate-based'' perspective of traditional definitions of topology (in which the system's eigenstates are used to calculate invariants, such as the Chern number \cite{klitzing_new_1980,thouless_quantized_1982,haldane_model_1988,haldane_possible_2008,raghu_analogs_2008}, Zak phase \cite{zak_berrys_1989}, or symmetry indicators \cite{benalcazar2017quad,benalcazar2017quadPRB,benalcazar_quantization_2019}). Nevertheless, the equivalence between the operator and eigenstate approaches can be understood by analyzing the properties of atomic limits. In an atomic limit, a system possesses a complete basis of spatially localized Wannier functions; in a crystal, these states form a flat band \cite{kitaev2009}. As such, from a real-space perspective, an atomic limit's Wannier states have both a well defined position and energy, which means that an atomic limit's Hamiltonian, $H^{(\textrm{AL})}$, commutes with its position operators, $X_j^{(\textrm{AL})}$, $[H^{(\textrm{AL})},X_j^{(\textrm{AL})}] = 0$. For systems that are not in an atomic limit, the spectral localizer establishes a system's topology by determining whether the system's Hamiltonian, $H$, and position operators, $X_j$, centered at some choice in position-energy space, $(\vec{x},E)$, can be continued to commuting without breaking any necessary symmetry or closing the local gap (i.e., does a continuous path of matrices $\{ X_j^{(\tau)},H^{(\tau)}\}$ with $0 \le \tau \le 1$ exist, where at every $\tau$ the necessary symmetries are preserved and the local gap is open, and in which $\tau = 0$ corresponds to the matrices of the original system centered at $(\vec{x},E)$ and $\tau = 1$ are those of the atomic limit). Any obstruction to this continuation yields a non-trivial local invariant and indicates that the system is topological. Overall, the spectral localizer's perspective on topology can be viewed as the real-space analogue to the perspective of topological quantum chemistry \cite{kruthoff_topological_2017,bradlyn_topological_2017,po_symmetry-based_2017,cano_building_2018,watanabe_space_2018,christensen_location_2022}, which uses a system's eigenstates and band structure to make a similar assessment of whether a set of bands below some chosen energy can be continued to an atomic limit, and where, again, any obstruction to this continuation is a manifestation of non-trivial topology.

To diagnose whether a finite, $d$-dimensional system at given position, $(x_1,\cdots,x_d)$, and energy, $E$, can be continued to an atomic limit, the system's operators are first shifted to be centered at that location, $X_j \rightarrow X_j - x_jI$ for $j \in 1,\ldots,d$ and $H \rightarrow H- EI$, where $I$ is the identity matrix. Then, to ascertain whether $H-EI$ and $X_j -x_jI$ can be continued to commuting, the spectral localizer combines these operators together using a non-trivial Clifford representation,
\begin{multline}
    L_{\boldsymbol{\lambda} = (x_1,\cdots,x_d,E)}(X_1,\cdots,X_{d},H)= \\
    \sum_{j=1}^{d} \kappa (X_{j}-x_{j} I)\otimes\Gamma_{j} + (H-E I)\otimes\Gamma_{d+1}.  \label{eq:loc}
\end{multline}
Here, $\Gamma_j^\dagger = \Gamma_j$, $\Gamma_j^2 = I$, and $\Gamma_j \Gamma_l = -\Gamma_l \Gamma_j$ for $j \ne l$, and $\kappa$ is a scaling parameter that ensures $H$ and $X_j$ have compatible units. Rigorously, one can prove that various properties of the spectral localizer, Eq.\ (\ref{eq:loc}), can be used to identify whether the set of matrices $\{(X_j-x_j I),(H - EI)\}$ has an obstruction that prohibits them from be continued to commuting (while preserving the necessary symmetries and local gap) for every symmetry class in every dimension that has the possibility to exhibit non-trivial topology \cite{loringPseuspectra}. However, just as different symmetry classes have different invariants in topological band theory, distinct properties of the spectral localizer are used for each symmetry class.

For example, in $2$ dimensions, the spectral localizer's invariant that determines whether the set of matrices $\{(X-x I),(Y-y I),(H - EI)\}$ can be continued to commuting while preserving their Hermiticity is
\begin{align}
    C_\textrm{L}(x,y,E) &= \tfrac{1}{2}\textrm{sig}\left( L_{(x,y,E)}(X,Y,H) \right) \in \mathbb{Z}, \label{eq:C}
\end{align}
in which $\textrm{sig}(L_{\boldsymbol{\lambda}})$ is a matrix's signature, its number of positive eigenvalues minus its number of negative eigenvalues; thus, $C_\textrm{L}(x,y,E)$ is guaranteed to be an integer. If $C_\textrm{L}(x,y,E) = 0$, the system is locally trivial and can be continued to an atomic limit. As this invariant does not take into account any other system symmetries (i.e., the system being described is in class A), $C_\textrm{L}(x,y,E)$ is a local Chern number. As a second example, in a $1$ dimensional system with particle-hole symmetry, $\mathcal{P}^2 = 1$ (i.e., class D), the invariant that identifies whether the set of matrices $\{(X-x I),H\}$ can be continued to commuting while preserving both particle-hole symmetry and their Hermiticity is
\begin{align}
    \tilde{\nu}_{\textrm{L}}\left(x,0\right) & = \textrm{sign}\left(\textrm{det}\left[ \left(\begin{array}{cc} 0 & I \end{array}\right)
    L_{\left(x,0\right)}(X,H)
    \left(\begin{array}{c} I \\ 0 \end{array}\right)\right]\right) \notag \\
    &= \textrm{sign}(\textrm{det}[(X-xI) + i H]) \in \{-1, 1 \} = \mathbb{Z}_2, \label{eq:D}
\end{align}
Here, the invariant is only defined at $E=0$, as particle-hole symmetry can only protect states at that energy, and the system is locally trivial if $\tilde{\nu}_{\textrm{L}}\left(x,0\right) = 1$. (Note, for class D, there is always a basis in which $H$ is purely imaginary and $X$ is real, so the determinant in Eq.\ (\ref{eq:D}) is guaranteed to be real.) 

For all symmetry classes, and in all dimensions, the local gap that Eq.\ (\ref{eq:loc}) preserves through the continuation process is 
\begin{equation}
    \mu_{\boldsymbol{\lambda}}^{\textrm{C}}(X_1,\cdots,X_{d},H) = \textrm{min}(|\sigma(L_{\boldsymbol{\lambda}}(X_1,\cdots,X_{d},H))|), \label{eq:muC}
\end{equation}
i.e., the absolute value of the eigenvalue of $L_{\boldsymbol{\lambda}}$ that is closest to zero. Here, $\sigma(L_{\boldsymbol{\lambda}})$ is the spectrum of $L_{\boldsymbol{\lambda}}$, and the superscript $\textrm{C}$ stands for Clifford, as this indicator function is related to the system's Clifford pseudospectrum \cite{cerjan_quadratic_2022}. None of the invariants that the spectral localizer uses to identify topology can be changed without $\mu_{\boldsymbol{\lambda}}^{\textrm{C}} \rightarrow 0$, as they are all continuous functions of invertible matrices that have the correct mathematical properties (e.g.\ the sign of the determinant is continuous on the set of invertible real matrices). Altogether, there are at least two ways that $\mu_{\boldsymbol{\lambda}}^{\textrm{C}}$ can close so that the topological invariant can change: by either changing one's choice of $\boldsymbol{\lambda}$, or by adding perturbations to the underlying operators, $X_j$ and $H$.

Note that the choice in position-energy space for where to evaluate the spectral localizer, $\boldsymbol{\lambda} = (x_1,\cdots,x_d,E) \in \mathbb{R}^{d+1}$, need not exist within the lattice's spatial or spectral extent. In other words, both the spectral localizer's invariants and local gap can be evaluated anywhere, and for any energy, regardless of the size of the finite system under consideration. This freedom of choice is analogous to the freedom in representation theory--based approaches to choose any number of bands to assess whether they are Wannierizable (i.e., whether that set of bands can be continued to the atomic limit). Just as adding or removing a band from a given set of bands can change their Wannerizability (e.g., different band gaps can have different topology), changing the choice of $\boldsymbol{\lambda}$ where the spectral localizer is evaluated can also affect whether the set of matrices $\{(X_j-x_j I),(H - EI)\}$ can be continued to commuting.

\subsection{Maxwell's equations as a Hermitian eigenproblem}

To apply the spectral localizer to photonic structures, the system must permit a description in terms of an effective Hamiltonian and position operators. Although there are some classes of photonic systems that can be approximated as tight-binding lattices \cite{rechtsman_photonic_2013,hafezi_imaging_2013,wimmer_experimental_2017} and could be immediately analyzed using Eq.\ (\ref{eq:loc}), here we seek a generic framework that is applicable to all photonic systems. Thus, in this section we will recast Maxwell's equations as a unique Hermitian eigenvalue problem and analyze the relationship between the symmetries of the structure and its effective Hamiltonian. To do so, we assume that all of the materials used in the system are linear, with spatially local responses, and that the fields possess a harmonic time dependence, $e^{-i\omega t}$. Under these conditions, Maxwell's source-free equations are
\begin{align}
&\nabla \times \vec{E}(\vec{x}) = i \omega \bar{\mu}(\vec{x},\omega)\vec{H}(\vec{x}), \label{eq:curlE} \\
&\nabla \times \vec{H}(\vec{x}) = -i \omega \bar{\varepsilon}(\vec{x},\omega) \vec{E}(\vec{x}), \label{eq:curlH} \\
&\nabla \cdot \left[\bar{\varepsilon}(\vec{x},\omega) \vec{E}(\vec{x}) \right] = 0, \label{eq:divE} \\
&\nabla \cdot \left[\bar{\mu}(\vec{x},\omega) \vec{H}(\vec{x}) \right] = 0. \label{eq:divH}
\end{align}
Here, $\mathbf{E}(\mathbf{x})$ and $\mathbf{H}(\mathbf{x})$ are the electric and magnetic fields, and $\bar{\varepsilon}(\mathbf{x},\omega)$ and $\bar{\mu}(\mathbf{x},\omega)$ are the spatially varying, possibly frequency-dependent, permittivity and permeability tensors of the system's constituent materials. Strictly speaking, it is not possible for a material to be both dispersive (i.e., possess a frequency dependent response) and completely lossless, as this violates the Kramers-Kronig relations \cite{jackson_classical_1998}. However, it is necessary for our framework to incorporate the possibility of dispersion, as many of the materials used in the construction of topological photonic systems are inherently dispersive (for example, magneto-optic materials that can be used to break time-reversal symmetry). Thus, we assume that any dispersive materials have narrow absorption lines that are sufficiently far away from the frequency ranges of interest, such that $\bar{\varepsilon}(\mathbf{x},\omega)$ and $\bar{\mu}(\mathbf{x},\omega)$ are approximately Hermitian within those frequency ranges. 

In contrast to other classes of physical systems, photonic systems are somewhat unusual as they generically possess two independent mechanisms through which they can dissipate energy: material absorption and radiation. Thus, even if all of a photonic system's constituent materials are energy-conserving, the system can still be rendered non-Hermitian by radiative boundary conditions \cite{sommerfeld_partial_1949,cerjan_why_2016}, which physically represent the loss of energy from a finite region due to radiation. As such, to obtain a Hermitian eigenvalue problem for a finite photonic system, we require that the system be bounded by a Hermitian boundary condition, such as periodic boundary conditions or a perfect electric conductor (PEC, i.e., Dirichlet boundary conditions on $\mathbf{E}(\mathbf{x})$). Then, for non-zero frequencies, Eqs.\ (\ref{eq:curlE}) and (\ref{eq:curlH}) form a self-consistent generalized Hermitian eigenvalue problem,
\begin{equation}
    W \boldsymbol{\uppsi}(\mathbf{x}) = \omega M(\vec{x},\omega) \boldsymbol{\uppsi}(\mathbf{x}), \label{eq:genEig}
\end{equation}
in which $\boldsymbol{\uppsi}(\mathbf{x}) = (\mathbf{H}(\mathbf{x}), \mathbf{E}(\mathbf{x}))^\trans$,
\begin{equation}
    W = \left( \begin{array}{cc}
   0 & -i \nabla \times \\
    i \nabla \times & 0
    \end{array} \right),
\end{equation}
and
\begin{equation}
    M(\vec{x},\omega) = \left( \begin{array}{cc}
    \bar{\mu}(\mathbf{x},\omega) & 0 \\
    0 & \bar{\varepsilon}(\mathbf{x},\omega) 
    \end{array} \right).
\end{equation}
Even though Eq.\ (\ref{eq:genEig}) only retains Eqs.\ (\ref{eq:curlE}) and (\ref{eq:curlH}) from Maxwell's equations, it maintains a complete description of the photonic system for any $\omega \ne 0$; Eqs.\ (\ref{eq:divE}) and (\ref{eq:divH}) can be recovered by taking the divergence of Eq.\ (\ref{eq:genEig}) and using the vector calculus identity $\nabla \cdot \nabla \times \vec{F}(\vec{x}) = 0$ for any vector field $\vec{F}(\vec{x})$. In general, solutions to Eq.\ (\ref{eq:genEig}) for dispersive materials can be found using iterative methods. However, it is also possible to remove the frequency dependence from $M(\vec{x},\omega)$ by adding auxiliary fields, and associated equations of motion for the material's internal degrees of freedom, to Eq.\ (\ref{eq:genEig}) \cite{raman_photonic_2010}; this allows for the generalized eigenproblem to be solved using standard methods at the cost of increasing the sizes of $\boldsymbol{\uppsi}$, $W$, and $M$.

To convert Eq.\ (\ref{eq:genEig}) into a unique ordinary Hermitian eigenvalue equation, we make the final assumption that $M(\vec{x},\omega)$ is positive semidefinite, at least over the frequency range of interest. Physically, this is not a significant restriction beyond the prior assumption that the constituent materials are energy-conserving, as it is effectively equivalent to requiring that the system's materials are dielectrics of some variety, possibly anisotropic or magneto-optic. Thus, for these frequencies of interest, $M(\vec{x},\omega)$ is guaranteed to possesses a unique, Hermitian, positive semidefinite square root matrix, $M^{1/2}(\vec{x},\omega)$. As such, by defining $\boldsymbol{\upphi}(\mathbf{x}) = M^{1/2}(\vec{x},\omega) \boldsymbol{\uppsi}(\mathbf{x})$, Eq.\ (\ref{eq:genEig}) can be written as
\begin{equation}
     H_{\textrm{eff}}(\vec{x},\omega) \boldsymbol{\upphi}(\mathbf{x}) = \omega \boldsymbol{\upphi}(\mathbf{x}), \label{eq:ordEig}
\end{equation}
where the system's effective Hamiltonian, 
\begin{equation}
H_{\textrm{eff}}(\vec{x},\omega) = M^{-1/2}(\vec{x},\omega) W M^{-1/2}(\vec{x},\omega) \label{eq:effHam}
\end{equation}
is both Hermitian and uniquely defined for every frequency. Note that even if a photonic system's constituent materials are non-dispersive (i.e., $M(\vec{x},\omega) = M(\vec{x})$), Gauss' laws (Eqs.\ (\ref{eq:divE}) and (\ref{eq:divH})) prohibit redefining the system's ``zero frequency'' to an arbitrary value as can be done in systems described by standard tight-binding models; $\omega = 0$ is a polarization singularity in Maxwell's equations where longitudinal modes appear \cite{christensen_location_2022}.

The determination of a material's topology is inextricably linked to its symmetries, both local \cite{schnyder2008,kitaev2009,ryu2010topological} and crystalline \cite{benalcazar2017quad,benalcazar2017quadPRB,kruthoff_topological_2017,bradlyn_topological_2017,po_symmetry-based_2017,cano_building_2018,benalcazar_quantization_2019}. In photonic systems, the presence or absence of a given symmetry typically manifests in its constituent materials and their spatial distribution, i.e., in $M(\vec{x},\omega)$. Thus, it is important to understand what happens to a symmetry of $M(\vec{x},\omega)$ when constructing $H_{\textrm{eff}}(\vec{x},\omega)$. Fortunately, one can use the Weierstrass approximation theorem \cite{Vaughn_Intro_math_phys} to prove that if $M(\vec{x},\omega)$ commutes or anti-commutes a unitary or anti-unitary symmetry, then $M^{-1/2}(\vec{x},\omega)$ possesses the same symmetry (see Supplementary Material for a proof), which greatly simplifies the analysis of the symmetries of $H_{\textrm{eff}}(\omega)$.


\subsection{Applying the spectral localizer to Maxwell's equations}

Given the unique effective Hamiltonian for Maxwell's equations, Eq.\ (\ref{eq:effHam}), coupled with an understanding of how its symmetries relate to those of the underlying photonic structure, the final step required to apply the spectral localizer, Eq.\ (\ref{eq:loc}), is to define the photonic system's position matrices. At present, the mathematics that underpins the spectral localizer's ability to assess whether a set of matrices can be continued to commuting must be applied to (arbitrarily large) finite matrices. Thus, one way to both construct position operators and ensure finite operators is to discretize the photonic system. This discretization can be performed using standard methods, such as the finite-difference Yee grid \cite{yee_numerical_1966}, or finite element methods \cite{monk_finite_2003}. Although any choice of discretization effectively imposes an (arbitrarily high) upper frequency cutoff to the photonic system's spectrum, the spectral localizer is provably local in both position and frequency. Thus, as the spectral localizer is insensitive to the system's details at frequencies sufficiently far away from the frequency range of interest \cite{raghu_analogs_2008}, invariants and gaps determined using it are guaranteed to converge.

Upon choosing a discretization scheme, we can rewrite $W$ and $M(\vec{x},\omega)$ as finite matrices that directly incorporate information about the system across its entire spatial extent. In particular, $M(\omega)$ can be expressed as a block diagonal matrix,
\begin{equation}
    M(\omega) = \left( \begin{array}{ccc}
        M(\vec{x}_1,\omega) &  & \\
         & M(\vec{x}_2,\omega) & \\
         & & \ddots
        \end{array} \right),
\end{equation}
in which each block represents the material properties at a particular vertex of the discretized system, while $W$ is a (usually sparse) matrix representing the curl operations and boundary conditions whose exact form will depend on the specific discretization scheme chosen. Note, as $M(\vec{x},\omega)$ is positive semidefinite, $M(\omega)$ is as well, and thus it also possesses a unique square root whose symmetries are directly given by those of the full structure,
\begin{multline}
    M(\omega)\mathcal{U} \pm \mathcal{U}M(\omega) = 0 \implies \\ 
    M^{-1/2}(\omega)\mathcal{U} \pm \mathcal{U}M^{-1/2}(\omega) = 0. \label{eq:sqrtMprops}
\end{multline}
where $\mathcal{U}$ is a unitary or anti-unitary operator (see Supplementary Material for a proof). Thus, altogether,
\begin{equation}
    H_{\textrm{eff}}(\omega) = M^{-1/2}(\omega) W M^{-1/2}(\omega) \label{eq:effHam2}
\end{equation}
is a unique $6n$-by-$6n$ matrix, where $n$ is the number of vertices in the discretization. Finally, in this basis, the position operators are simply the coordinates of the vertices of the discretization scheme.

For the remainder of this study, we make use of a finite-difference Yee grid discretization.

\section{Photonic Chern insulator}

\begin{figure*}[th]
    \centering
    \includegraphics{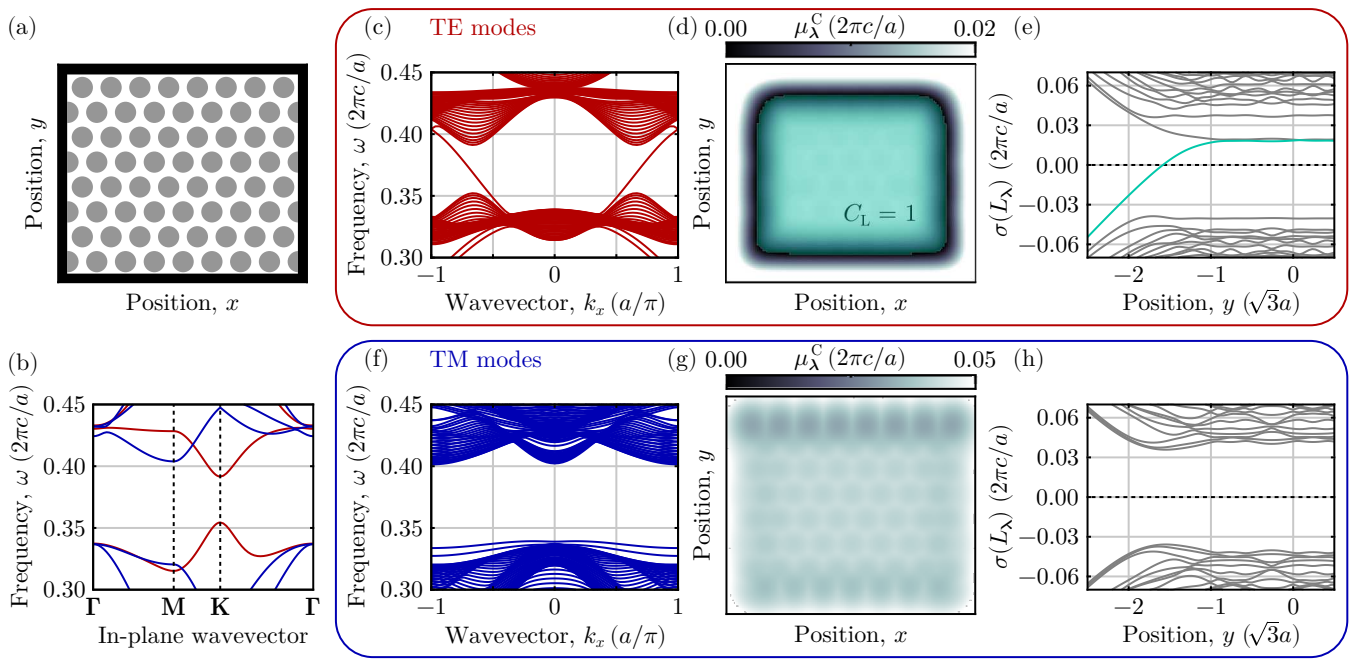}
    \caption{(a) Schematic of a 2D photonic crystal consisting of high-dielectric rods, $\varepsilon_{jj} = 14$ for $j=x,y,z$, with radius $r = 0.37a$ embedded in a gyro-electric background, $\varepsilon_{jj} = 1$ and $\varepsilon_{xy} = -0.4i$. The lattice constant is $a$ and all materials have $\mu_{jj} = 1$. The black boundary indicates where perfect electric conductor boundary conditions were used. (b) Bulk band structure for both the TE (red) and TM (blue) bands. (c) Ribbon band structure for the TE fields with PEC boundaries along the two non-periodic edges. (d) Spatially resolved local gap, $\mu_{\boldsymbol{\lambda}}^{\textrm{C}}(X,Y,H_{\textrm{eff}})$, at $\boldsymbol{\lambda} = (x, y, \omega = 0.37(2\pi c/a))$ with $\kappa = 0.25$. The gap is in units of $(2\pi c/a)$. The topological region of the system with $C_L = 1$ is shown in lime, and the system's scale is identical to (a). (e) Full localizer spectrum at $\boldsymbol{\lambda} = (x =0, y, \omega = 0.37(2\pi c/a))$. The eigenvalue which is responsible for the change of the signature of $L_{\boldsymbol{\lambda}}$ is highlighted in lime.
    (f)-(h) Similar to (c)-(e), except for the TM fields.
      \label{fig:haldane}}
\end{figure*}

To provide a concrete example of how the topology of a photonic crystal can be determined using the spectral localizer directly from Maxwell's equations, we first consider the well-known photonic Chern insulator proposed by Haldane and Raghu \cite{haldane_possible_2008,raghu_analogs_2008}. This system consists of a triangular lattice of high-dielectric rods embedded in a low-dielectric gyro-electric background (Fig.\ \ref{fig:haldane}a), and in which the frequency-dependence of the gyro-electric response has been ignored. In the absence of any time-reversal symmetry breaking (i.e., the external magnetic field is turned off), this system exhibits a Dirac point at $\mathbf{K}$ in its Brillouin zone for its transverse electric (TE) modes that coincides with a complete band gap in the transverse magnetic (TM) modes. As the strength of the time-reversal symmetry breaking is increased, a topological band gap opens in the TE modes (Fig.\ \ref{fig:haldane}b), yielding boundary-localized chiral edge TE states that can be seen in a ribbon band structure (Fig.\ \ref{fig:haldane}c).

As a direct comparison, we show how the spectral localizer reveals the topology of this photonic Chern insulator without calculating its band structure or eigenstates. In a 2D photonic structure, the spectral localizer can be explicitly written as
\begin{multline}
    L_{\boldsymbol{\lambda} = (x,y,\omega)}(X,Y,H_{\textrm{eff}}) = \\
        \left( \begin{array}{cc}
        H_{\textrm{eff}} - \omega I & \kappa(X-xI) - i\kappa(Y-yI) \\
        \kappa(X-xI) + i\kappa(Y-yI) & -(H_{\textrm{eff}} - \omega I)
        \end{array} \right). \label{eq:loc2d}
\end{multline}
Here, Maxwell's equations are directly incorporated through the definition of $H_{\textrm{eff}}$, Eq.\ (\ref{eq:effHam2}), and we have used PEC boundary conditions to ensure the system is finite. In the system's bulk, we find that the local Chern number for the TE modes, Eq.\ (\ref{eq:C}), is non-trivial, $C_\textrm{L}(x,y,\omega) = 1$, while beyond the system's boundaries the system has trivial topology, $C_\textrm{L}(x,y,\omega) = 0$ (Figs.\ \ref{fig:haldane}d,e). Thus, as the system's topology must switch between these two domains, the local gap must close at the domain boundary, $\mu_{\boldsymbol{\lambda}}^{\textrm{C}} = 0$, which approximately coincides with the system's physical boundaries. The closing of the local gap is a direct manifestation of bulk-boundary correspondence in the system, and indicates the presence of boundary-localized photonic chiral edge states. In contrast, the same quantities for the system's TM modes show that this modal sector is topologically trivial within the same complete band gap, regardless of whether $(x,y)$ are chosen within or outside the system's bulk (Figs.\ \ref{fig:haldane}f-h).

Beyond qualitative agreement, the topology predicted by the spectral localizer demonstrates quantitative agreement with the band theoretic calculation: the system's local gaps in both polarization sectors agree with its respective bulk band gaps. The complete band gap for the TE modes has a width of approximately $\Delta \omega \approx 0.04 (2\pi c/a)$. Using the spectral localizer to calculate the local gap at the middle of the TE band gap ($\omega = 0.37 (2\pi c/a)$), we find that $\mu_{\boldsymbol{\lambda}}^{\textrm{C}} \approx 0.02 (2\pi c/a)$ in the photonic crystal's interior. Thus, by probing the system at the center of its bulk band gap, we find that the local gap is approximately half of the bulk band gap, i.e., these two measures of topological protection are in nearly exact agreement (from the probed central frequency, a shift of half the bulk band gap is necessary to reach the nearest bulk band edge, which is what $\mu_{\boldsymbol{\lambda}}^{\textrm{C}}$ measures in a gapped system's bulk). For the TM modes, the result is similar, with the spectral localizer yielding a local gap, $\mu_{\boldsymbol{\lambda}}^{\textrm{C}} \approx 0.04 (2\pi c/a)$, in the system's bulk that is a bit larger than the frequency distance between $\omega = 0.37 (2\pi c/a)$ and the lower frequency bound of the nearest bulk TM band, $\omega \approx 0.40 (2\pi c/a)$. Given the rigorous connection between the size of the local gap and the system's topological protection \cite{loringPseuspectra}, the larger TM local gap indicates that, in terms of adding a perturbation to the photonic system, it is more difficult than the bulk band gap suggests to change the TM sector to possess non-trivial topology.

A note on implementation --- it is not necessary to calculate the full spectrum of $L_{\boldsymbol{\lambda}}$ to find its signature, and performing the calculation this way may be prohibitively computationally expensive for many photonic systems. Instead, due to Sylvester's Law of Inertia \cite{sylvester_xix_1852,higham2014sylvester}, one can first find the LDLT decomposition, $L_{\boldsymbol{\lambda}} = N D N^\dagger$, and then $\textrm{sig}(L_{\boldsymbol{\lambda}}) = \textrm{sig}(D)$, where $D$ is diagonal. Thus, as there are commonly available LDLT decomposition methods for sparse matrices, $C_\textrm{L}$ can be efficiently computed.

\section{Topological phases in gapless photonic crystals}

Separating a system's measure of topological protection from its bulk band gap enables the definition of topology in gapless systems \cite{cerjan_local_2021}, which is of particular importance in photonic systems as it is experimentally challenging to realize photonic Chern insulators using magneto-optic materials \cite{bahari_nonreciprocal_2017,bahari_photonic_2021}. Specifically, as known magneto-optic materials provide relatively modest changes to a system's dielectric tensor at technologically relevant wavelengths \cite{espinola_magneto-optical_2004,bi_-chip_2011}, the design of photonic Chern insulators is currently burdened by requiring the system's geometry to serve two separate functions: maximizing the size of the topological band gap while maintaining a complete band gap at the same frequency. Instead, in this section, we show that the ability to characterize the topology of systems even in the absence of a bulk band gap enables the design of photonic structures that can focus solely on maximizing the size of a local gap, potentially expanding the frequency range over which a system can be proven to exhibit topological behaviors. Although any resulting boundary-localized resonances can hybridize with the available degenerate bulk states, their partial edge localization and reduced ability to back-scatter \cite{bergman_bulk_2010,junck_transport_2013} may still be useful for enhancing light-matter interactions, such as lasing \cite{bahari_nonreciprocal_2017,st-jean_lasing_2017,bandres_topological_2018,zeng_electrically_2020,yang_spin-momentum-locked_2020,shao_high-performance_2020,bahari_photonic_2021,dikopoltsev_topological_2021,yang_topological-cavity_2022} or soliton formation \cite{lumer_self-localized_2013,ablowitz_linear_2014,leykam_edge_2016,zhou_optical_2017,mukherjee_observation_2020,maczewsky_nonlinearity-induced_2020,leykam_probing_2021,mittal_topological_2021,jurgensen_quantized_2021,jurgensen_chern_2022,fu_nonlinear_2022,maluckov_nonlinear_2022}.

\begin{figure}[t!]
    \centering
    \includegraphics[width=1.0\columnwidth]{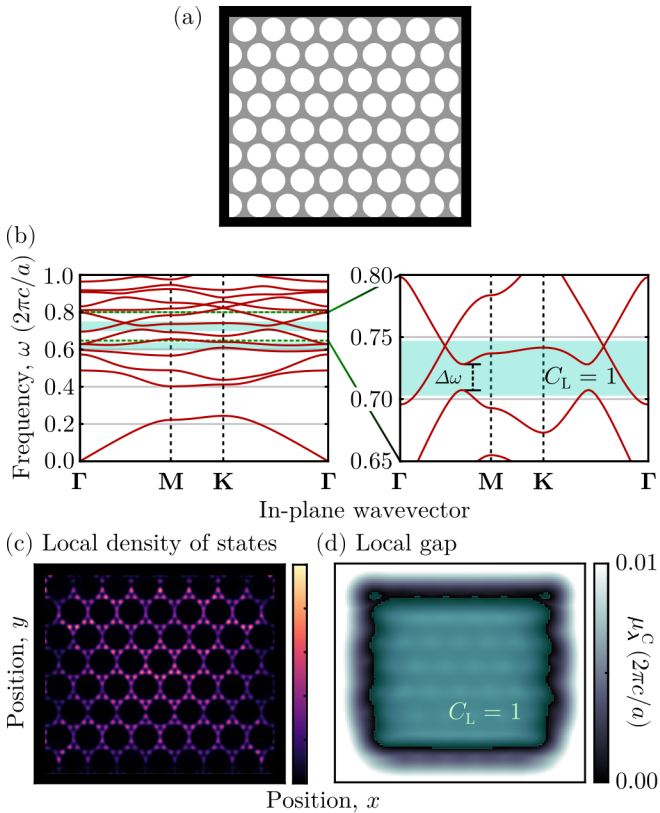}
    \caption{(a) Schematic of a 2D photonic crystal of air holes with radius $r = 0.4a$ embedded in a high-dielectric background, $\varepsilon_{xx} = \varepsilon_{yy} = 12.25$. Both the air and background are approximated to be frequency-independent gyro-electric materials, $\varepsilon_{xy} = -0.24i$. The lattice constant is $a$ and all materials have $\mu = 1$.  The black boundary indicates where perfect electric conductor boundary conditions were used. (b) Bulk band structure of this photonic crystal's TE modes, lime-colored regions indicate frequency ranges where $C_{\textrm{L}} = 1$. The zoomed in right panel indicates one of the topological frequency ranges, with the local band gap, $\Delta \omega$, indicated. (c) Local density of states for the $H_z$ component of the fields for $\omega = 0.720 (2 \pi c/a)$. (d) 2D local gap and local Chern number for $\omega = 0.720 (2 \pi c/a)$ and $\kappa = 0.10$. For both (c) and (d), the spatial scale is the same as the system shown in (a).
      \label{fig:gapless}}
\end{figure}

The crucial realization that enables the prediction of gapless topological photonic systems is that, although a bulk band gap guarantees a similarly sized local gap \cite{LoringSchuBa_even,LoringSchuBa_odd}, it is possible for $\mu_{\boldsymbol{\lambda}}^{\textrm{C}} > 0$ even in the absence of a bulk band gap. Thus, in regions with local gaps, the spectral localizer can be used to predict robust, non-trivial topology and associated phenomena. Moreover, such gapless topological systems still possess a bulk-boundary correspondence as the local gap is required to close, $\mu_{\boldsymbol{\lambda}}^{\textrm{C}} \rightarrow 0$, across the system's boundary for the local topology to change.
An example of a gapless gyro-electric photonic crystal with non-trivial topology is illustrated in Fig.\ \ref{fig:gapless}, where we simulate the TE modes of a 2D photonic crystal consisting of a triangular lattice of air holes embedded in a high-dielectric background. We artificially consider both the air and high-dielectric materials to have a magneto-optic response, which serves as an approximation of the experimental realization of these systems through layering a patterned slab of high-dielectric material on top of a solid slab of magneto-optic material \cite{bahari_nonreciprocal_2017,bahari_photonic_2021}. As can be proven using the 2D spectral localizer, Eq.\ (\ref{eq:loc2d}), this system possesses a pair of frequency ranges with non-trivial topology that can be identified by the local Chern number, $C_\textrm{L}$ (see Fig.\ \ref{fig:gapless}b), despite the fact that \textit{neither} of these topological frequency ranges coincide with complete band gaps of the system's TE modes. Moreover, even though the presence of bulk states obscures the identification of any boundary-localized resonances in the system's local density of states (see Fig.\ \ref{fig:gapless}c), we can explicitly confirm the presence and approximate locations of this system's chiral edge resonances in these topological frequency ranges by determining where its local gap vanishes, Fig.\ \ref{fig:gapless}d. Quantitatively, this calculation also shows that for $\omega$ in the middle of the topological frequency range, the local gap within the system, $\mu_{\boldsymbol{\lambda}}^{\textrm{C}} \approx 0.01 (2 \pi c / a)$, is approximately half of the incomplete topological band gap, $\Delta \omega \approx 0.02 (2 \pi c / a)$ (marked in Fig.\ \ref{fig:gapless}b), indicating that this topological phase is protected against disorder that does not close the incomplete band gap, regardless of the absence of a complete band gap.

Beyond reducing the design constraints for developing topological photonic crystals, the ability to identify topological frequency ranges without requiring a bulk band gap enables the discovery and design of topological photonic crystal systems that operate at high normalized frequencies, where they are unlikely to possess a complete band gap. As the normalized frequencies shown in Fig.\ \ref{fig:gapless} are equal to the ratio of the lattice constant, $a$, to the operating wavelength, $\lambda$, i.e., $\omega a / (2\pi c) = a / \lambda$, using phenomena that appear at higher normalized frequencies (for a fixed operating wavelength) improves these systems' fabrication tolerances by increasing the system's lattice constant. 

\section{An operator-based topological crystalline invariant}

Recently, there has been significant interest in photonic systems that exhibit topological states protected by crystalline invariants, as these structures can exhibit robust waveguide- and cavity-like states for enhancing light-matter interactions \cite{ota_active_2020} and do not require materials or configurations that break time-reversal symmetry. However, as crystalline symmetries fall outside of the standard classification of topological systems \cite{schnyder2008,kitaev2009,ryu2010topological}, invariants that rely on crystalline symmetries for predicting topological behaviors must be calculated using a separate theoretical framework, such as through Wannier centers \cite{kruthoff_topological_2017,bradlyn_topological_2017,po_symmetry-based_2017,cano_building_2018} or symmetry indicator invariants \cite{benalcazar2017quad,benalcazar2017quadPRB,benalcazar_quantization_2019}. In this section, we show how to incorporate crystalline symmetries directly into the operator-based framework of the spectral localizer, and we provide an example of such a topological crystalline invariant. Not only does this yield a theory for topological crystalline systems that is not dependent on a system's Bloch eigenstates, but it also shows how crystalline symmetries can be placed on equal footing with those symmetries considered in the ten-fold classification scheme.

The key mathematical observation that allows for the spectral localizer to consider topological crystalline structures is that the $C^*$-algebra theorems that underpin its operator-based framework \cite{Moutuou_GradedBrauerGroups,altland1997nonstandard,loringPseuspectra} are agnostic to the specific physical interpretation of the operators. Thus, any local topological invariant found using the spectral localizer can be repurposed for any set of symmetries, so long as those symmetries obey similar relations with system's operators, $\{X_1,\cdots,X_d,H \}$. This argument is best illustrated using an example. Consider a 1D system with chiral symmetry (i.e., class AIII), such that the chiral operator, $\Pi$, anti-commutes with the system's Hamiltonian, $H\Pi = - \Pi H$, and commutes with its position operator, $X \Pi = \Pi X$ (in a lattice-vertex basis, both $\Pi$ and $X$ are diagonal, so this commutation relation is guaranteed in general). The topology of such a system can be determined both using traditional methods, such as a winding number \cite{asboth2016short_course_top_ins}, or using the spectral localizer's associated invariant,
\begin{align}
    \nu_{\textrm{L}}\left(x,0\right) & = \tfrac{1}{2}\text{sig}\left[ \left(\begin{array}{cc} 0 & I \end{array}\right)
    L_{\left(x,0\right)}(X,H)
    \left(\begin{array}{c} \Pi \\ 0 \end{array}\right)\right] \notag \\
    & = \tfrac{1}{2}\text{sig}\left[ (\kappa(X - xI) + i H) \Pi \right] \in \mathbb{Z}, \label{eq:indAIII}
\end{align}
which is only well-defined at $E=0$, as chiral symmetry can only protect states at mid-gap. However, any other unitary operator, $\mathcal{U}$, that satisfies the same set of relations with the system's operators, $\mathcal{U}H = \mp H\mathcal{U}$ and $\mathcal{U}X = \pm X\mathcal{U}$, will possess a nearly identical invariant as a chiral symmetric system, with $\Pi \rightarrow \mathcal{U}$ in Eq.\ (\ref{eq:indAIII}) and possibly also $X \leftrightarrow H$ and $x \leftrightarrow E$ depending on whether $H$ or $X$ anti-commutes with $\mathcal{U}$.

In particular, this argument for re-purposing existing operator-based invariants for crystalline symmetries can be immediately applied to determine the topology of inversion symmetric systems, as the inversion operator, $\mathcal{I}$, satisfies $\mathcal{I}H = H\mathcal{I}$ and $\mathcal{I}X = - X\mathcal{I}$. To demonstrate the versatility and generality of this method, we study an inversion-symmetric photonic system consisting of a bipartite array of air holes in a high-dielectric background that possesses an inversion center (Fig.\ \ref{fig:inv}a), similar to designs used in photonic nanobeams \cite{ohta_strong_2011,burek_high_2014,sipahigil_integrated_2016}. Due to the staggered spacing of the air holes, this system exhibits a bulk band gap in its TE modes, in which a defect-localized state appears that is bound to the inversion center (Figs.\ \ref{fig:inv}b,c). When the air holes are evenly spaced, no such localized state appears within the system's lowest bands (Figs.\ \ref{fig:inv}d,e).

\begin{figure}[t!]
    \centering
    \includegraphics[width=1.0\columnwidth]{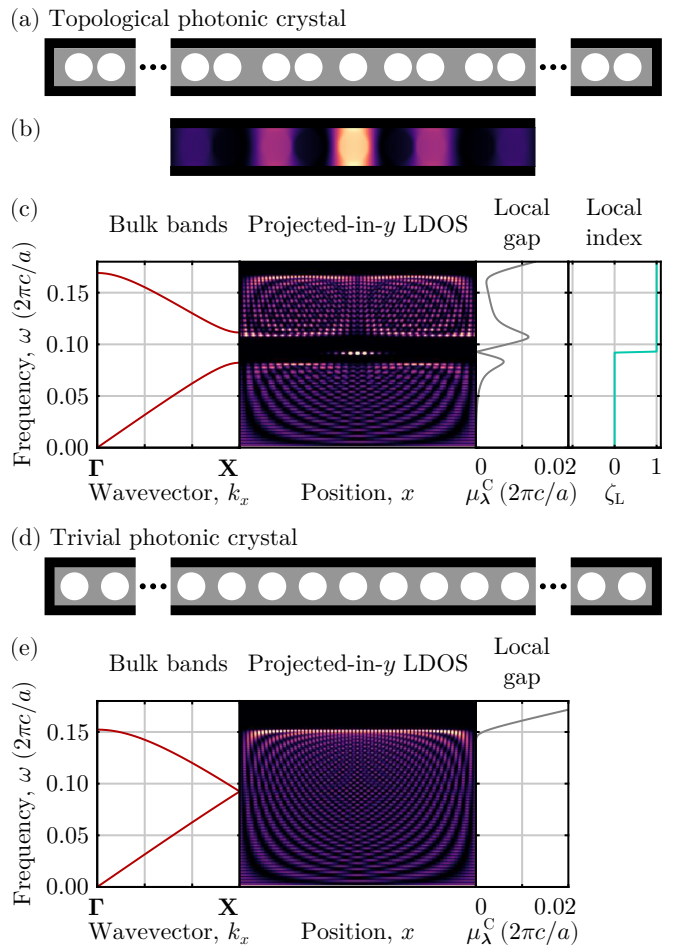}
    \caption{(a) Schematic of a 1D photonic crystal consisting of air holes with radius $r = 0.35a$ embedded in a high-dielectric background, $\varepsilon = 14$ and $\mu = 1$, where $a$ is the lattice constant of the uniform system. Spacing between neighboring air holes is alternately increased and decreased by $0.2a$. The full system consists of an inversion-symmetric defect and 16 pairs of air holes on each side. The black boundary indicates where perfect electric conductor boundary conditions were used. (b) Local density of states for the $H_z$ component of the fields for $\omega = 0.091 (2 \pi c/a)$ shown on the same scale as (a). (c) (From left to right) Bulk band structure, projected-in-$y$ local density of states (LDOS), local gap (see Eq.\ \ref{eq:muC}), and local index (see Eq.\ \ref{eq:zeta}) for this finite photonic system. The local gap and index are calculated at $x = 0$ using $\kappa = 0.01$. The local gap is given in units of $(2\pi c/a)$. (d),(e) Similar to (a) and (c), but for the uniform system, with all air holes separated by $a$. The local index is not shown in (e) to emphasize its lack of meaning in the absence of a local gap.
      \label{fig:inv}}
\end{figure}

To prove that this defect-localized state is of topological origin, we construct the system's 1D spectral localizer, which can be explicitly written as,
\begin{align}
    &L_{\boldsymbol{\lambda} = (x,\omega)}(X,H_{\textrm{eff}}) =
        \left( \begin{array}{cc}
        0 & A \\
        A^\dagger & 0
        \end{array} \right), \label{eq:loc1d} \\
    &A = \kappa(X-xI) - i(H_{\textrm{eff}}-\omega I). \notag
\end{align}
Although the photonic system is not uniform in $y$, here we are purposefully omitting $Y-yI$ in Eq.\ (\ref{eq:loc1d}), which has the effect of projecting the entire system onto the $x$-axis within the spectral localizer's framework. By doing so, we can re-purpose the 1D local winding number, Eq.\ (\ref{eq:indAIII}), for inversion symmetric systems to yield the crystalline invariant,
\begin{equation}
    \zeta_{\textrm{L}}\left(0,\omega \right) = \tfrac{1}{2}\text{sig}\left[ (H_{\textrm{eff}} - \omega I + i \kappa X) \mathcal{I} \right] \in \mathbb{Z}. \label{eq:zeta}
\end{equation}
Just as the winding number can only protect states in chiral symmetric systems at mid-gap, $\zeta_{\textrm{L}}$ can only protect states at the inversion center, $x = 0$, but these states can have any frequency, $\omega$. Using this invariant, we observe that the defect-localized state in the bipartite photonic crystal is topological, as this state's appearance coincides with a shift in $\zeta_{\textrm{L}}$ (i.e., there is a bulk-boundary correspondence), and the state's frequency is protected by the large local gap, $\mu_{\boldsymbol{\lambda}}^{\textrm{C}}$, that appears at both immediately higher and lower frequencies (right panels of Fig.\ \ref{fig:inv}c). In contrast, the uniform photonic crystal is topologically trivial, and completely lacks a local gap within the frequency range of the first TE band (right panel of Fig.\ \ref{fig:inv}e). (For completeness, we note that an analogous invariant to $\zeta_{\textrm{L}}$ could be constructed using $Y$ in place of $X$, which also anti-commutes with $\mathcal{I}$. However, for the systems in Fig.\ \ref{fig:inv}, this $Y$-based invariant is always trivial.)

There are a few points that are worth emphasizing for the spectral localizer's topological crystalline invariants. First, unlike previous crystalline invariants for inversion-symmetric systems, such as the Zak phase \cite{zak_berrys_1989} and symmetry indicators \cite{benalcazar_quantization_2019}, $\zeta_{\textrm{L}}$ is a $\mathbb{Z}$ invariant, not a $\mathbb{Z}_2$ invariant, and thus it can identify topology in systems that may be mis-identified as trivial by these previous invariants. Second, the standard winding number for 1D systems and the inversion invariant in Eq.\ (\ref{eq:zeta}) are both manifestations of the same type of $K$-theory; specifically, they determine elements in some $K$-theory group of a graded real $C^*$-algebra of matrices and operators that respect symmetries induced by two anti-unitary operators. (See the appendix of \cite{Moutuou_GradedBrauerGroups} how such algebras arise as a mathematical version of the ten-fold way \cite{altland1997nonstandard}.) Finally, as the spectral localizer is provably local \cite{loring_guide_2019,cerjan_quadratic_2022}, a system need not be perfectly inversion symmetric across its entire extent --- beyond some window whose width is related to $\Vert[H,\kappa X]\Vert$, perturbations to $H$ and $X$ away from a chosen frequency and position cannot meaningfully effect the spectral localizer's properties. Thus, in practice, a system need not be globally inversion symmetric to exhibit topological states protected by inversion symmetry, only locally so.

\section{Discussion}

In conclusion, we have developed an operator-based framework for determining a photonic structure's topology using the spectral localizer. As this theory is based entirely on the system's real-space description, it is immediately applicable to aperiodic and disordered structures which do not possess a band structure or Bloch eigenstates. Moreover, using this framework, we have shown two developments for topological photonic systems. First, by leveraging the spectral localizer's ability to define a measure of topological protection separate from a system's bulk band gap, we have shown that it is possible to find robust topological states even in gapless photonic systems. This development has potentially significant experimental implications, as some previous implementations of topological lasers that do not exhibit large bulk band gaps may, in fact, possess more topological protection than their bulk band gaps suggest \cite{bahari_nonreciprocal_2017,bahari_photonic_2021}. Second, as the mathematical theorems which underpin the spectral localizer's framework are agnostic to the specific physical meaning of any of the system's operators, the spectral localizer's invariants can be immediately re-purposed to handle crystalline symmetries. This development shows how crystalline symmetries can be placed on equally strong footing to the topology of ``local'' symmetries considered in the ten-fold classification of lattices \cite{schnyder2008,kitaev2009,ryu2010topological}. Moreover, our framework can be used to determine both the topology of photonic topological crystalline insulators, and the strength of the protection of any localized states, without the need to calculate symmetry indicators or Wannier centers \cite{benalcazar2017quad,benalcazar2017quadPRB,kruthoff_topological_2017,bradlyn_topological_2017,po_symmetry-based_2017,cano_building_2018,benalcazar_quantization_2019}. We anticipate this development will substantially increase the possible design space for developing new topological cavity-like states in photonic systems. (Our framework also avoids issues associated with the polarization singularity at zero frequency and momentum \cite{christensen_location_2022}.) Looking forward, as our framework is not reliant upon a system's band structure, it may offer the possibility of yielding a general theory of topology in non-linear photonic systems.

\begin{acknowledgments}
T.L. acknowledges support from the National Science Foundation, grant DMS-2110398.
A.C. and T.L. acknowledge support from the Center for Integrated Nanotechnologies, an Office of Science User Facility operated for the U.S.\ Department of Energy (DOE) Office of Science, and the Laboratory Directed Research and Development program at Sandia National Laboratories. Sandia National Laboratories is a multimission laboratory managed and operated by National Technology \& Engineering Solutions of Sandia, LLC, a wholly owned subsidiary of Honeywell International, Inc., for the U.S.\ DOE's National Nuclear Security Administration under contract DE-NA-0003525. The views expressed in the article do not necessarily represent the views of the U.S.\ DOE or the United States Government.
\end{acknowledgments}

\section*{Supplemental Material}
The supplementary material contains proofs for Eq.\ (\ref{eq:sqrtMprops}) and Refs.\ \cite{ConwayFunctionalAnalysis,LorindVides_Floquet_with_Symmetry,schulz2021invariants,weierstrass1885analytische}.

\bibliography{pseudo_for_maxwell}

\end{document}


\title{Supplemental Material: An operator-based approach to topological photonics}

\author{Alexander Cerjan}
\email{awcerja@sandia.gov}
\affiliation{Center for Integrated Nanotechnologies, Sandia National Laboratories, Albuquerque, New Mexico 87185, USA}

\author{Terry A.\ Loring}
\email{loring@math.unm.edu}
\affiliation{Department of Mathematics and Statistics, University of New Mexico, Albuquerque, New Mexico 87131, USA}

\date{\today}

\maketitle

\section*{Matrix square roots and symmetries}

In the main text we relied on the fact that if a positive semi-definite matrix $M$ satisfies a symmetry then so must $M^\frac{1}{2} $ and $M^{-\frac{1}{2}}$.  Indeed, this claim is also true for bounded operators on a Hilbert space.  Note that in operator theory \cite{ConwayFunctionalAnalysis}, the term positive is more common than positive semi-definite.  We will prove things for the finite matrix case, where eigenvalues are guaranteed and the spectral theorem is simpler to state.

If $M$ is positive semidefinite, then by the spectral theorem we can factor $M$ as 
\begin{equation}
M=UDU^\dagger \label{eqn:spectral_decomp}
\end{equation}
with $U$ unitary and $D$ diagonal with the diagonal elements $\lambda_j = D_{j,j}$ all in the closed, finite interval $[0,\|D\|]$.

Given a positive semidefinite matrix $M$ there will be  a unique positive
semidefinite matrix $N$ so that $N^2=M$.  Of all the many square roots of $M$, it is only $N$ that is anointed with the notation $M^\frac{1}{2}$.  A good way to understand $N$ is via Eq.~(\ref{eqn:spectral_decomp}).  Given that factorization, one quickly finds (see \cite[Ex.~2.16.]{Vaughn_Intro_math_phys}
\begin{equation}
M^\frac{1}{2}=UD_1U^\dagger \label{eqn:spectral_decomp_half}
\end{equation}
where $D_1$ is again diagonal, but with $j$th diagonal element $\sqrt{\lambda_j}$.
This way of calculating the square root makes it hard to track the effect of a symmetry, so we look to one of the many alternate means of calculating $M$ by hand.  We are not discussing numerical methods of computing matrix square roots that respect symmetries, but that is of interest \cite[\S~IV]{LorindVides_Floquet_with_Symmetry}. 

Given a sequence of polynomials $p_n$, if we select this so that we have
good convergence $p_n(\lambda)\rightarrow \sqrt{\lambda}$  then we will have  \begin{equation}
M^\frac{1}{2} = \lim_{n\rightarrow \infty} p_n(M). \label{eqn:root_by_limit_of_poly}
 \end{equation}
To be technical, we are applying the Weierstrass Approximation Theorem \cite[\S~6.2.1]{Vaughn_Intro_math_phys} which guarantees a sequence of  polynomials that uniformly converge on the interval  $[0,\|M\|]$.  The advantage of this approach is that with a polynomial $p(\lambda) = \sum a_n \lambda ^n$ we need not use Eq.~(\ref{eqn:spectral_decomp_half}). Instead, we can work with the more direct formula
\begin{equation}
p(M) = \sum a_n M^n . 
\end{equation}

We get now to our first result.  It is not new, but in the math literature one will generally not see any antiunitary operators, and will instead find results about elements of real $C^*$-algebras \cite{schulz2021invariants}.  As such, it is simpler to provide direct proofs rather than to translate between the two pictures.

\begin{lem}
Suppose $\mathcal{U}:\mathbb{C}^n \rightarrow \mathbb{C}^n$ is a unitary or antiunitary operator.  Suppose also that $M$ is a positive semidefinite $n$-by-$n$ matrix that we treat as a linear operator on $\mathbb{C}^n$.  If $\mathcal{U}$ commutes with $M$ then $\mathcal{U}$ commutes with $M^\frac{1}{2} $.
\end{lem}

\begin{proof}
We are assuming $\mathcal{U}\circ M = M \circ  \mathcal{U}$ and it immediately follows that 
\begin{align*}
\mathcal{U} \circ M \circ \cdots \circ M  &=  M \circ \mathcal{U} \circ  \cdots \circ M \\
         & \qquad \qquad \vdots \\
         & =   M \circ \cdots \circ M  \circ \mathcal{U}  
\end{align*}
which says $\mathcal{U} \circ M^n = M^n \circ \mathcal{U}$. For the case of unitary $\mathcal{U}$, everything proceeds as expected,
\begin{equation}
\mathcal{U} \circ \left( \sum a_n M^n \right)=  \left(\sum a_n M^n \right) \circ \mathcal{U}.
\end{equation}
Unfortunately, for antiunitary $\mathcal{U}$ we find
\begin{equation}
\mathcal{U} \circ \left( \sum a_n M^n \right)=  \left(\sum \overline{a_n} M^n \right) \circ \mathcal{U}.
\end{equation}
However, fortunately, we are able to select the $p_n$
to have only real coefficients, by the the original theorem of Weierstrass \cite{weierstrass1885analytische}, so we find  $\mathcal{U} \circ p_n(M) =  p_n(M) \circ \mathcal{U} $ for both the unitary and antiunitary cases.
Finally, both unitary and antiunitary operators are continuous, so we pass to the limit and are done
by Eq.~(\ref{eqn:root_by_limit_of_poly})
\end{proof}

To be fair, we are jumping over some complications when we say
we are passing to the limit.  What we need is at least 
point-wise convergence for
\begin{equation*}
  \lim_{n\rightarrow\infty}\left(\mathcal{U}\circ p_{n}(M)\right)=\mathcal{U}\circ\left(\lim_{n\rightarrow\infty}p_{n}(M)\right)  
\end{equation*}
and
\begin{equation*}
   \lim_{n\rightarrow\infty}\left(p_{n}(M)\circ\mathcal{U}\right)=\left(\lim_{n\rightarrow\infty}p_{n}(M)\right)\circ\mathcal{U}. 
\end{equation*}
In the unitary case these are standard results.  The antiunitary versions have simple, but not short, proofs.  At the core of these proofs is the fact that the antiunitary property
\begin{equation*}
    \left\langle \mathcal{U}(\boldsymbol{v}),\mathcal{U}(\boldsymbol{w})\right\rangle =\overline{\left\langle \boldsymbol{v},\boldsymbol{w}\right\rangle }
\end{equation*}
implies that $\mathcal{U}$ preserves distances:
\begin{equation*}
    \|\mathcal{U}(\boldsymbol{x})-\mathcal{U}(\boldsymbol{y})\|=\|\boldsymbol{x}-\boldsymbol{y}\|.
\end{equation*}

\begin{lem}
Suppose $\mathcal{U}:\mathbb{C}^n \rightarrow \mathbb{C}^n$ is a unitary or antiunitary operator. Suppose also that $M$ is a positive semidefinite $n$-by-$n$ matrix.  If $\mathcal{U}$ anti-commutes with $M$ then $\mathcal{U}$ anti-commutes with $M^\frac{1}{2} $.
\end{lem}

\begin{proof}
The proof is essentially as before, except that we can only prove that  $\mathcal{U} \circ M^n = - M^n \circ \mathcal{U}$ for odd values of $n$.  However, we can approximate the odd function
\begin{equation}
\lambda\mapsto\begin{cases}
\sqrt{\lambda} & \text{for }\lambda\geq0\\
-\sqrt{-\lambda} & \text{for }\lambda\leq0
\end{cases}
\end{equation}
on the larger interval  $[-\|M\|,\|M\|]$  by polynomials, and this means we can zero out all the coefficients in even positions.  The proof then
proceeds as before.
\end{proof}

\begin{lem}
Suppose $\mathcal{U}:\mathbb{C}^n \rightarrow \mathbb{C}^n$ is a unitary or antiunitary operator. Suppose also that $M$ is an invertible $n$-by-$n$ matrix.  If $\mathcal{U}$ commutes with $M$ then $\mathcal{U}$   commutes with $M^{-1} $.
\end{lem}

\begin{proof}
\begin{equation}
M\circ\mathcal{U}=\mathcal{U}\circ M\circ\mathcal{U}\circ M^{-1}
\end{equation}
\begin{equation}
M^{-1}\circ M\circ\mathcal{U}\circ M^{-1}=M^{-1}\circ\mathcal{U}\circ M\circ M^{-1}
\end{equation}
which simplifies to 
\begin{equation}
\mathcal{U}\circ M^{-1}=M^{-1}\circ\mathcal{U}
\end{equation}
\end{proof}

The proof is easily adapted to prove the following.

\begin{lem}
Suppose $\mathcal{U}:\mathbb{C}^n \rightarrow \mathbb{C}^n$ is a unitary or antiunitary operator. Suppose also that $M$ is an invertible $n$-by-$n$ matrix.  If $\mathcal{U}$ anti-commutes with $M$ then $\mathcal{U}$ anti-commutes with $M^{-1} $.
\end{lem}

\bibliography{pseudo_for_maxwell}